%% file: 0_Private_SDN.tex
\documentclass[10pt,conference,letterpaper]{IEEEtran}

\usepackage{setspace}
\usepackage[cmex10]{amsmath}
\usepackage{amssymb}
\usepackage{latexsym}
\usepackage[english]{babel}
\usepackage{amsmath}
\usepackage{verbatim}
\usepackage{tabularx}
\usepackage{graphicx}
\usepackage{algorithm}
\usepackage{algorithmic}
\usepackage{subfigure}
\usepackage{threeparttable}
\newtheorem{define}{Definition}
\newtheorem{proposition}[define]{Proposition}

\usepackage{multirow}

\newenvironment{proof1}{
 \removelastskip\par\medskip
 \removelastskip\par
 \noindent{\em Proof Sketch:}
  \rm}{\penalty-20\null\hfill$\Box$\par\medbreak}

\title{Privacy-preserving Cross-domain Routing Optimization
 --A Cryptographic Approach}

\author{\IEEEauthorblockN{Qingjun Chen}
\IEEEauthorblockA{
Nanjing University\\
qingjunchen@smail.nju.edu.cn}
\and
\IEEEauthorblockN{Chen Qian}
\IEEEauthorblockA{
University of Kentucky\\
qian@cs.uky.edu}
\and
\IEEEauthorblockN{Sheng Zhong}
\IEEEauthorblockA{
Nanjing University\\
zhongsheng@nju.edu.cn }}

\begin{document}

\maketitle
\thispagestyle{plain}
\pagestyle{plain}

\begin{abstract}
Today's large-scale enterprise networks,  data center networks, and wide area networks can be decomposed into multiple administrative or geographical domains. Domains may be owned by different administrative units or organizations. Hence protecting domain information is an important concern. Existing general-purpose Secure Multi-Party Computation (SMPC) methods that preserves privacy for domains are extremely slow  for cross-domain routing problems.
In this paper we present PYCRO, a cryptographic protocol  specifically designed for  privacy-preserving cross-domain routing optimization in Software Defined Networking (SDN) environments. PYCRO provides two fundamental routing functions,  policy-compliant shortest path computing and bandwidth allocation, while ensuring strong protection for the private information of domains. We rigorously prove the privacy guarantee of our protocol. We have implemented a prototype system that runs PYCRO on servers in a campus network. Experimental results using real ISP network topologies show that PYCRO is very efficient in computation and communication costs.
\end{abstract}

\input{1_introduction}

\input{2_relatedwork}
\input{3_background}
\input{4_baseline}
\input{5_optimized}
\input{7_capacity}

\input{6_analysis}

\input{8_evaluation}

\vspace{-1ex}
\section{Conclusion}
\label{sec:conclusion}
In this paper we present PYCRO, the first privacy-preserving cross-domain routing optimization protocol in SDN environments. We develop a new cryptographic tool named the Secure-If operation and apply it with homomorphic encryption to compute the shortest cross-domain paths without revealing private information. PYCRO also provides bandwidth allocation, a fundamental traffic engineering solution. We have implemented PYCRO in a prototype system and  performed real experiments to demonstrate its efficiency. Experimental results show that PYCRO can improve the time and bandwidth efficiency by an order of magnitude compared to general-purpose solutions. In future we will design more complex routing optimization functions based on PYCRO.  We believe our study may lead to useful discussion of the same problem for the Internet.

\input{0_Private_SDN.bbl}
\end{document}

%% file: 1_introduction.tex
\section{Introduction}
\label{introduction}

Large-scale enterprise networks,  data center networks, and wide area networks (WANs) may be decomposed  into multiple administrative or geographical  domains \cite{sepia2010, ROME, B4, MS-SDN, SDNLocality, Disco, SDNarchitecture}. Multi-domain networks are deployed to interconnect community networks, data centers, corporation  sites,  and university campuses. In a multi-domain network such as a WAN, different domains may belong to different administrative units with an organization or different organizations \cite{ROME, B4, MS-SDN, Disco, FlowBroker, SDNarchitecture}. For example, a number of organizations may own their own sub-networks, and those subnetworks are mutually interconnected to form a multi-domain network \cite{SDNarchitecture}. Hence individual domain may have security and privacy concerns regarding revealing its domain information to other domains. \emph{This paper focuses on multi-domain networks that consist of a relatively small number of domains, which may appear in current enterprise networks and WANs. We do not consider Internet-scale multi-domain networks at this stage. }

Routing optimization, such as finding policy-compliant paths that have least routing cost or satisfy bandwidth demands, plays a critical role of network management.
Recent advances of Software Defined Networking (SDN) has brought tremendous convenience to routing optimization by separating the control plane from routers and allowing a central controller to make routing decisions. Using centralized optimization, the controller can efficiently and effectively find a desired routing path for each flow and install forwarding rules on corresponding switches.\footnote{In this work we refer all network units as ``switches'' for consistency to SDN terminology. }
Although SDN simplifies routing optimization in a single domain, privacy-preserving cross-domain routing optimization is still a challenging problem.
Suppose each domain has a centralized controller. The state-of-the-art approach to route a cross-domain flow is using local optimization for intra-domain path selection and  BGP for inter-domain routing, such as the design in Google's SDN B4 \cite{B4} and DISCO \cite{Disco}. This approach protects the autonomy and privacy of domains. However, it is obvious that local optimization plus BGP may not find an network-wide optimized path and can hardly provide bandwidth guarantee. Another solution is to allow every controller to broadcast its domain information to the entire network and maintains a network-wide map, similar to a controller-level OSPF protocol. This approach causes privacy and security concerns because every domain has to expose its private information  such as network topology, link latencies, bandwidth, and routing policies. In fact, there is no practical and privacy-preserving solution to the most fundamental routing problem, i.e., computing shortest paths, for multi-domain networks.

Privacy-preserving cross-domain network problems can be modeled as secure multi-party computation (SMPC) \cite{yao1982,fairplaymp2008,
 sepia2010,
huang2011faster,
lindell2012secure,nielsen2012new}. However, general-purpose SMPC solutions such as SEPIA \cite{sepia2010} are extremely slow and may take days to complete \cite{connectivity} \cite{SMPC-BGP}. Therefore, customized algorithms are needed for the privacy-preserving cross-domain routing problems.

In this paper, we present the first work for  privacy-preserving cross-domain routing optimization that has reasonable efficiency in practical networks. We design and implement a protocol named PYCRO and its extensions to provide two fundamental routing functions, namely policy-compliant shortest path computing and bandwidth allocation, while protecting the private information of domains. PYCRO is executed on SDN controllers in a distributed manner and does not rely on any trusted third party.
PYCRO is developed based on a novel cryptographic tool called Secure-If operations.

The properties of PYCRO can be summarized as follows.

1) PYCRO can compute policy-compliant cross-domain shortest paths and allocate bandwidth for flows while protecting private information of domains. The privacy guarantee of PYCRO is cryptographically strong.(Please see Section~\ref{sec:analysis} for formal analysis of privacy.)

2) PYCRO also preserves the autonomy and local policies of domains. A domain can independently  determine whether and how to forward different flows and these preferences are unknown to other domains.

3) PYCRO is efficient in both computation and communication costs.

PYCRO is the first work of privacy-preserving cross-domain routing optimization in SDN environments. We have implemented a prototype system that runs PYCRO on machines in a campus network. Experimental results using real ISP network topologies show that PYCRO has reasonably good efficiency. It spends $<30$ seconds and $< 700$ KB messages in computing a shortest path tree for  networks consisting of thousands of switches and links.

The rest of this paper is organized as follows. We review the related work in Section~\ref{sec:relatedwork}. In Section~\ref{problem formulation}, we introduce the problem overview and background. We presenting the PYCRO protocol in Section~\ref{sec:design}, and then introduce some optimization techniques in Section \ref{sec:opti}. We design the bandwidth allocation protocol with our PYCRO protocol in Section~\ref{sec:ba}. In Section~\ref{sec:analysis}, we justify the privacy-preserving property of  PYCRO. We evaluate the performance of our protocol in Section~\ref{sec:evaluation}.  Finally, we conclude our paper in Section~\ref{sec:conclusion}.

%% file: 2_relatedwork.tex
\section{Related Work}
\label{sec:relatedwork}
Privacy-preserving cross-domain routing can be modeled as a secure multi-party computation (SMPC) problem.   Yao's seminal work~\cite{yao1982} introduces the first algorithm, called Yao's garbled circuits, to allow two parties to compute an arbitrary function with their inputs without revealing private information. Since then, many studies about SMPC have been conducted
\cite{fairplaymp2008,
 sepia2010,
huang2011faster,
lindell2012secure,nielsen2012new}. In~\cite{fairplay2004}, a secure two-party computation system called Fairplay is introduced and the system implements generic secure function evaluation. FairplayMP proposed in~\cite{fairplaymp2008} supplements the Fairplay system. FairplayMP is a generic system for secure multi-party computation while Fairplay only supports secure two-party computation. SEPIA \cite{sepia2010} is a recently proposed SMPC system for general inter-domain network applications.
A common limitation of these SMPC  solutions is that the computation time can be way too long for practical applications. For example, \cite{connectivity} shows that it takes thousands of days to track cross-domain connectivity of a few domains using SEPIA \cite{sepia2010}. An SMPC-based routing algorithm proposed to replace BGP also experiences long execution time \cite{SMPC-BGP} which makes the existing SMPC methods impractical for inter-domain routing.

Recently researchers have proposed custom privacy-preserving algorithms for different network applications. Chen \emph{et al. } \cite{cross-reachability} use Bloom filters to combine access control lists of multiple domains and determine network reachability in a privacy-preserving manner. Djatmiko \emph{et al. } \cite{connectivity} propose to apply counting Bloom filters for privacy-preserving multi-domain connectivity tracking. STRIP \cite{PPvector} is a privacy-preserving inter-domain routing protocol to replace BGP and achieve fast convergence. To our knowledge, no existing work in this category studies the privacy-preserving cross-domain routing optimization problem.

%% file: 3_background.tex
\section{Problem Overview and Background}
\label{problem formulation}

In the section, we formalize the problem in this paper and then introduce a novel cryptographic tool we will use to solve the problem.

\subsection{Problem Formulation}
We formalize the problem to be solve in this paper as follows.

Consider a large network that consists of $N$ domains: $D_1$, $D_2$, \ldots, $D_N$, where each domain
$D_i$ has a \emph{domain controller} $C_i$ that makes routing decision and updates the forwarding tables of switches in the domain.
A domain controller can access any information of its domain, including the network topology, access control policies, link bandwidth, and authenticated hosts. A domain controller can add, delete, and update  forwarding entries of switches in its domain. It communicates with controllers in other domains via pre-established secure channels.

For any two switches
$v, v' \in D_i$ ($v \neq v'$), we use $v \sim v'$ to denote that there is a link between $v$ and $v'$ and we denote its link cost by $c(vv')$.
Clearly, each $C_i$ should know the topology of $D_i$, and should also know all the link costs within this domain: $\{c(vv')| v, v' \in D_i, v \neq v'\}$.
We assume that the intra-domain topology and the intra-domain link costs are all private information of $C_i$. That is, any other domain controller should not know anything about this topology or these link costs.
We assume different domains are managed by different parties, such as ISPs, organizations, or departments of a corporation. Parties do not share domain information. If a party owns multiple physical domains, all these domains can be considered a single logical domain in this problem.

There are some inter-domain links that connect switches from different domains. We assume that information about an inter-domain link is available of the two end domains, and domains can share it with other domains. That is, for any inter-domain link $vv'$ (where
$v \in D_i$, $v' \in D_j$ and $D_i \neq D_j$), all domain controllers could know the two endpoints $v$ and
$v'$, and also $D_i$ and $D_j$---the domains they belong to. A switch that is connected to switches in other domains is called a \emph{gateway switch}. We assume gateway switches are publicly known.

Suppose that there are a source switch $v_s$, which belongs to a domain $D_s$, and a destination switch $v_t$, which belongs to another domain $D_t$. Our objective is
to design a private-preserving optimized routing solution. Specifically,  we need
to design a protocol that
allows each domain controller $C_i$ to find the forwarding table $T(v)$ for all $v \in D_i$,
where each entry $T(v)[v_s, v_t]$
is the next-hop switch of $v$ on the optimal routing path from the source  $v_s$ to
destination $v_t$.

In this work, PYCRO focuses on two major routing optimization problems.

 1) \emph{Policy-compliant shortest path routing}. Each link has an associated routing cost (also known as link weight), representing a performance metric such as hop count, latency, or traffic load \cite{linkstateTE}.
    The routing object is to find a  path from the source to the destination that has the minimum sum of link cost without violating policies of domains.

2) \emph{Bandwidth allocation.} Bandwidth allocation has been applied to practical traffic engineering solutions such as B4 \cite{B4}. Each flow has a bandwidth demand and link bandwidth is allocated to different flows. When flows are competing for bandwidth, a single flow may need multiple paths to satisfy its bandwidth demand. The routing object is to find one or more paths for a flow such that the flow bandwidth demand can be satisfied. At this stage, we do not consider fairness among flows \cite{B4}.


\textbf{Security and Privacy Requirements.}
Due to security concerns, a switch only allows its domain controller to install, delete, or update forwarding table entries.
Domains may not wish to reveal their information including network topology, link bandwidth, and routing policies. In addition, a domain should have routing autonomy to determine whether and how to forward a given flow. This preference should also be made confidential to other domains.



\subsection{Cryptographic Tool}
Here we introduce the cryptographic tool we will use in this work, namely the Secure-If operation.

\textbf{Secure-If operation. }
Our protocol depends on a cryptographic technique developed by us, which we call \emph{the Secure-If operation}.
This operation allows the protocol to choose between two options $Y$ and $Z$, based on whether a particular condition $X$ is satisfied. Denote by $SecIf(X,Y,Z)$ the Secure-If operation, and then we have
\begin{equation}
SecIf[X,Y,Z] = \begin{cases}
Y, & X~is~satisfied; \\
Z, & otherwise.
\end{cases}
\end{equation}
Note that this operation is privacy preserving. It is infeasible for anybody to decide whether the condition is satisfied or not, i.e., which of the two options is actually chosen.
For example, suppose
that $X$, $Y$, $Z$ are ciphertexts; consider a condition that ``$X$ is an encrypted $1$''.
This operation can return a rerandomization of $Y$ when
the plaintext of $X$ is indeed $1$, and return a rerandomization of $Z$ otherwise.
However, nobody can learn whether the returned value is a rerandomization of $Y$ or a rerandomization of $Z$ unless the result is decrypted.
In general, the privacy guarantee is that no knowledge about any plaintext(s) involved is leaked
to any party.

The involved conditions may be complicated and thus this technique itself can
depend on other cryptographic building blocks. For instance, we may need to
use the building block of \emph{partial decryption}. Suppose that the private key
for a ciphertext
is shared among a number of parties using a secret sharing scheme \cite{secretsharing}. Partial
decryption allows a party with a share of the private key to partially decrypt a
ciphertext. The partially decrypted ciphertext does not leak any knowledge about
the plaintext. However, when a threshold number of parties apply
partial decryption one by one, the plaintext will finally be revealed.
Detailed implementation of Secure-If operations are custom-built and depend on different algorithms. 

Also notice that we will use a few variants of this technique in this paper. Each of these variants is constructed in a distinct way. Please see Section~\ref{sec:if} for the detailed constructions.

%% file: 4_baseline.tex
\section{Design of the PYCRO Protocol}
\label{sec:design}
In this section, we present the PYCRO protocol with three steps:equivalent cost graph construction, privacy-preserving shortest path tree protocol and path establishment. In the PYCRO protocol,
we need two homomorphic encryption systems $E()$ and $E'()$, both of which must be
\emph{semantically secure}.
The difference between $E()$ and $E'()$ is that $E()$ must be additively homomorphic, while
$E'()$ must be multiplicative homomorphic.
Specifically, for two messages $x$ and $y$,
\[E(x)+E(y)=E(x+y)\]
\[E'(x)\cdot E'(y)=E'(xy)\]
All $E()$ and $E'()$ encryption operations in this paper use a public key whose corresponding
private key is shared among the domain
controllers using $(N,2)$-secret sharing.
 There exist cryptosystems \cite{gentry2009fully,van2010fully,smart2010fully} that are both additively and multiplicatively homomorphic. However, we do not use them due to efficiency considerations.
We denote by $D()$ and $D'()$ the corresponding decryption operations,
respectively.
In addition, we allow both of them supports \emph{re-randomization operations}, and the rerandomization operation is denoted by $R()$ and $R'()$.
%
As mentioned earlier, another main cryptographic tool we use in the PYCRO protocol is the Secure-If operation.

\begin{figure*}[t]
\centering
\begin{tabular}{p{160pt}p{160pt}p{160pt}}
\subfigure[Equivalent cost graph of four domains: dashed lines are intra-domain links and solid lines are inter-domain links.]
{
\includegraphics[width=5.3cm]{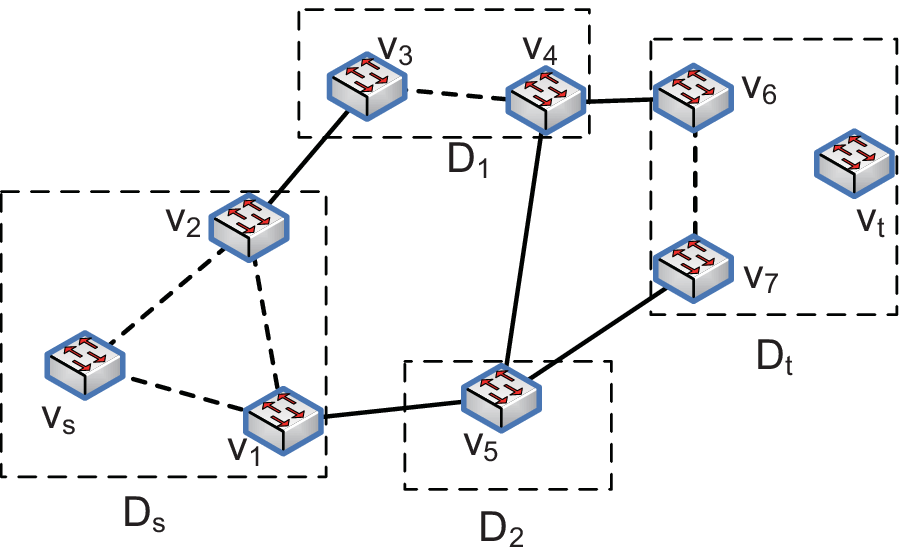}
\label{fig:ECG}
}&
\subfigure[An iteration of PSPT construction]
{
\includegraphics[width=5.3cm]{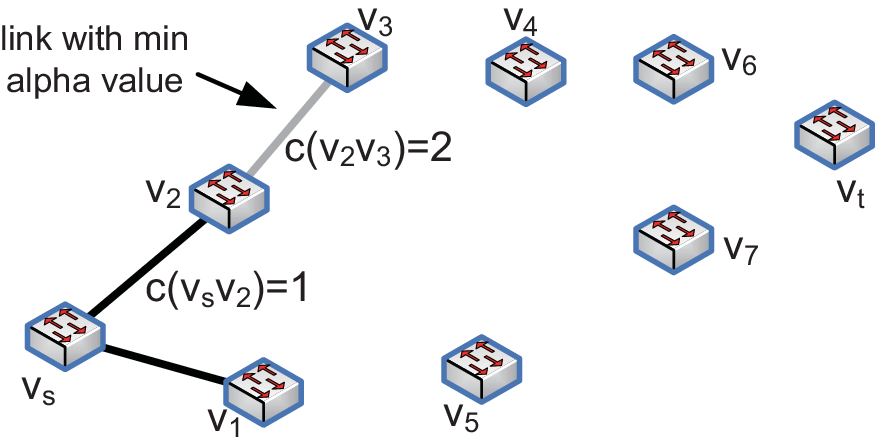}
\label{fig:ADD}
}&
\subfigure[PSPT and path establishment]
{
\includegraphics[width=5.3cm]{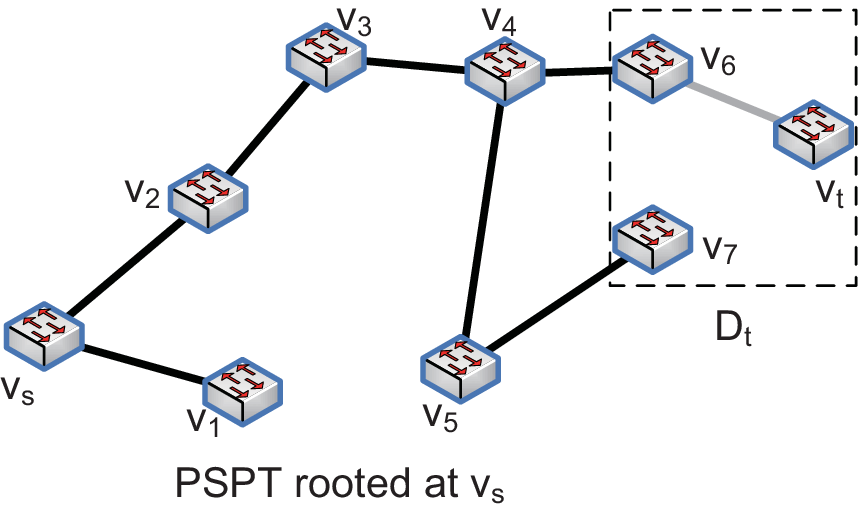}
\label{fig:Tree}
}
 \end{tabular}
\caption{An illustration of the PYCRO protocol.}
\vspace{-2ex}
\end{figure*}

\subsection{Equivalent Cost Graph Construction}

In this subsetion, we show how to construct the equivalent cost graph.
To construct equivalent cost graph, we first show the nodes in it and then the links in it.

As for nodes, we define a switch as a \emph{significant node} if it is the source switch or a gateway switch 
and the nodes of the equivalent cost graph are the significant nodes in the entire network.
We denote by $S_i$ the significant node set of domain $D_i$ and we also denote by $S$ the significant node set of all domains.

As for links, for any two significant nodes $v$ and $v'$ ($v \neq v'$), we distinguish two cases:

\emph{Case 1:} If $v,v' \in S_i$, then link $v \sim v'$ is in the equivalent
cost graph.
In this case, the link is called \emph{intra-domain link} since two nodes are in the same domain. Note that a intra-domain link does not necessarily correspond to a physical link, and could be a multi-hop path between two switches. The path from $v$ to $v'$ is selected by $D_i$ in the best effort based on $D_i$'s local policies and is not necessarily the shortest path. If a domain does not wish to forward $f$, it sets the path length as infinity or the pre-defined path length upper limit. We use $d(vv')$ to denote the path length assigned by $D_i$.

\emph{Case 2:} If $v \in S_i \wedge v' \in S_j \wedge S_i \neq S_j \wedge v \sim v'$
then link $v \sim v'$ is in the equivalent cost graph.
In this case, the link is called \emph{inter-domain link} since two nodes are in different domains. We use $c(vv')$ to denote the length of link $v \sim v'$.

As an example, Figure \ref{fig:ECG} shows the equivalent cost graph of a network consisting of four domains, in the view of the controller $C_s$ of the source domain $D_s$. The nodes of the graph are the source switch $v_s$ and all gateway switches $v_{1-7}$.

Clearly, $C_s$, the controller of the source domain, knows the connectivity information of the equivalent cost graph.
Furthermore, for links in Case 2 above, $C_s$ also knows the link costs in the equivalent cost graph.
For links in Case 1 above that are not in $D_s$, $C_s$ does
not know the link costs in the equivalent cost graph, which are private information of different domains.


\subsection{Privacy-preserving Shortest Path Tree Protocol}
\label{sec:tree}
This subsection describes how the source controller computes a Privacy-preserving Shortest Path Tree (PSPT) on the equivalent cost graph rooted at $v_s$ 
while providing strong protection for the private information of other domains.
We use $c_{max}$ to denote the maximum
link cost and assume the length of cryptographic keys in use is much greater than the length of $c_{max}$.

Each domain controller $C_i$, except the source domain controller $ C_s$, encrypts all its link costs in Case 1 of the equivalent cost graph, and sends them to $C_s$. Specifically, for any two switches $v$ and $v'$ in  $D_i$ ($D_i \neq D_s$), $C_i$ computes $e(vv')=E(d(vv'))$ and sends it to $C_s$. 
The source domain controller $C_s$ needs to encrypts all its link costs in Case 1 of the equivalent cost graph.
$C_s$ is also responsible for encrypting the link costs in Case 2 of the equivalent cost graph. Specifically, for
any $v$ in $D_i$ and  $v'$ in  $D_j$,  if there is an inter-domain link between these two nodes, then $C_i$ computes $e(vv')=E(c(vv'))$. For any $v, v' \in D_s$ ($v \neq v'$),  if there is an intra-domain link between these two nodes, then $C_i$ computes $e(vv')=E(c(vv'))$.

For each node $v$ in the equivalent cost graph, except the source node
itself, $C_s$ computes three indicators:
$f(v)=E'(2)$, $g(v)=E(0)$, and $h(v)=E'(\phi)$. Here $f(v)$ is an encrypted
indicator for node $v$, indicating whether it has been added to the shorted path
tree. We use an encrypted $2$ to represent ``No'', and an encrypted $2^{-1}$ to
represent ``Yes''.  The plaintext of $g(v)$ will be used for the length of the shortest
path from the source node to $v$, once $v$ is added to the shortest path tree.
The plaintext of $h(v)$ will be used to store the information of the parent node
of $v$ in the shortest path tree, once $v$ is added to the shortest path tree. All these indicators are essential in the computation of the shortest path tree.

For the source node, $C_s$ computes the three indicators: $f(v)=E'(2^{-1})$,
$g(v)=E(0)$, and $h(v)=E'(\phi)$, because it is the root of the tree.
Then the source controller repeats the two steps below for $|S|-1$ iterations, where $S$ is the set of nodes in the equivalent cost graph. At each iteration, a node with the minimum distance to the root among the remaining nodes is added to the tree.

Step 1. For each link $vv'$  in the equivalent cost graph,  $C_s$ uses a Secure-If
operation (denoted as $SecIf_0$) to compute $\alpha(vv')$.
The condition here is that the plaintext of $f(v)$ is equal to the plaintext of $f(v')$.
If this condition is satisfied,
$\alpha(vv')= E(c_{max}|S|+1)$; otherwise,
$\alpha(vv')=R(g(v)+g(v')+e(vv')))$.  If the condition is satisfied, it means either $v'$ and $v$ are both in the tree or neither in the tree. We just let $\alpha(vv')$ be a maximum value and do not consider it. If the condition is not satisfied, one of $v'$ and $v$ is in the tree and the other is not.  Then the plaintext of $\alpha(vv')$ is the distance from the node not in the tree to the root.

Step 2. For each link $vv'$  in the equivalent cost graph,  $C_s$ uses a
Secure-If operation (denoted as $SecIf_1$) to re-compute $f(v)$, $f(v')$, $g(v)$, $g(v')$, $h(v)$, $h(v')$.
The condition is that the plaintext of $\alpha(vv')$ is the smallest among the
$\alpha$ values of all links in the equivalent cost graph. The node not in the tree that corresponds to the smallest $\alpha$ should be added to the tree and its three indicators should be updated.

If the condition in Step 2 is satisfied, then we use another Secure-If operation (denoted as $SecIf_2$)
to decide how to update the indicators. The condition of this new Secure-If
operation is that the plaintext of $f(v)$ is equal to $2$, i.e., whether the node $v$ is not in the tree.
\begin{itemize}
\item When the condition is satisfied ($v$ is not in the tree),  $f(v)=E'(2^{-1})$, $g(v)=R(\alpha(vv'))$,
$h(v)=E'(v')$. The indicators of $v'$ are re-randomized.
\item Otherwise, $v'$ is not in the tree, hence $f(v')=E'(2^{-1})$, $g(v')=R(\alpha(vv'))$,
$h(v')=E'(v)$. The indicators of $v$ are re-randomized.
\end{itemize}

If the condition in Step 2 is not satisfied, all indicators $f(v)$, $f(v')$, $g(v)$, $g(v')$, $h(v)$, $h(v')$
are just re-randomized based on the original values.

We show an example of the above iteration in Figure \ref{fig:ADD}. $v_s$, $v_1$, and $v_2$ are already in the tree. Since $v_3$ is not in the tree, we compute $\alpha(v_2 v_3)=R(g(v_2)+g(v_3)+e(v_2 v_3)))=R(E(1)+E(0)+E(2))=R(E(3))$. Suppose the plaintext of $\alpha(v_2 v_3)$, i.e., 3, is the smallest $\alpha$ value. Then $v_3$ should be added to the tree. The indicators of $v_3$ are updated as follows: $f(v_3)=E'(2^{-1})$, $g(v_3)=R(E(3))$, $h(v)=E'(v_2)$. The indicators of $v_2$ are all re-randomized.

The detailed algorithm specification of the PSPT construction  protocol is not shown due to space limit. 
Once the algorithm is completed, for each $v$ in the equivalent cost graph, $C_s$ actually obtains the ciphertexts of $g(v)$, the shortest path length from $v_s$ to $v$, and $h(v)$, the parent of $v$ on the PSPT. Figure \ref{fig:Tree} shows the constructed PSPT of the network.
With all the $g(v)$ and $h(v)$, we can construct the path from $v_s$ to $v_t$ using the method proposed in the next section (Section~\ref{sec:path}). Note that we use three types of Secure-If operations ($SecIf_0$, $SecIf_1$ and $SecIf_2$). We will describe how they are implemented in detail in Section \ref{sec:if}.

\subsection{Path Establishment}
\label{sec:path}

After running the PSPT construction protocol, each domain controller knows all its significant nodes' values of $g$ and $h$ from $C_s$. Using the values, we can construct the path $P$ from $v_t$ back to $v_s$ step by step (e.g., first $v_t$, and then the parent of $v_t$, and then the parent of the parent of $v_t$, until the source $v_s$).


After finishing computing the shortest path tree, $C_s$ then partially decrypts each $g(v)$ and each $h(v)$, and
sends the partial decrypted ciphertexts to the domain controller of
node $v$. The domain controller of $v$ also applies partial decryption,
and thus obtains the plaintexts of $g(v)$ and $h(v)$, i.e., $dg(v)$ and $dh(v)$. Since $E()$ uses $(N,2)$-secret sharing, the encrypted indicators can be decrypted by partial decryption of two domains.

For any destination $v_t$, the shortest path and corresponding forwarding table entries are constructed using 
Algorithm~\ref{alg:path_establishment} 
with all plaintext indicators $dg()$ and $dh()$.

If $v_t$ is not a significant node, the domain controller $C_t$ of $v_t$ compares all the significant
nodes in its domain, for the sums of their distances from $v_s$ and to $v_t$. Suppose
the significant node with the smallest distance sum is $v$. Then the intra-domain
path from $v$ to $v_t$ is chosen as part of the shortest path from $v_s$ to $v_t$, and
the forwarding table entries for destination $v_t$ in this part of path are computed and installed accordingly.
The forwarding table entries in the other parts of the path are computed in a way similar to
$v_t$ being a significant node presented below.

If $v_t$ is a significant node,  the domain controller of $v_t$ decides what to do based on
the type of link between $v_t$'s parent $dh(v_t)$ on the shortest path tree  and $v_t$ in the equivalent cost graph.
\begin{itemize}
\item
If the link represents an intra-domain path, i.e., $dh(v_t)$ is another significant node in the destination domain,
the intra-domain path between $dh(v_t)$ and $v_t$  is picked as part of the shortest path from
$v_s$ to $v_t$.  The forwarding table entries for destination $v_t$ in the destination domain are installed by $C_t$ accordingly.
\item If the link is an inter-domain link, the link is added directly as part of the shortest path from
$v_s$ to $v_t$.
 $C_t$ then sends a message to the domain controller of $dh(v_t)$ and asks it to install a corresponding forwarding table entry at switch $dh(v_t)$.
\end{itemize}
Next, the domain controller of the predecessor of the destination domain on the selected shortest path computes the forwarding table entries similarly. This
process is repeated until the source switch is reached and all forwarding table entries for destination
$v_t$ have been computed.

For the example of Figure \ref{fig:Tree}, the destination controller $C_t$ selects $v_6$ as part of the optimal path from $v_s$ to $v_t$. It then installs forwarding entries on switches between $v_t$ and $v_6$ and also notifies $C_1$ to install a forwarding table entry at $v_4$, specifying that packets from $v_s$ to $v_t$ should be forwarded to $v_6$ by $v_4$. The routing path can be established by repeating this process.

Algorithm~\ref{alg:path_establishment} presents the pseudocode of the path establishment protocol in Section \ref{sec:path}.

\begin{algorithm}[!ht]
\caption{Path Establishment Protocol} \label{alg:path_establishment}
\begin{algorithmic}[1]
\REQUIRE All significant nodes' $g$ and $h$;\\
		 Source node $v_s$ and destination node $v_t$;
\ENSURE The shortest path $P$ from $v_s$ to $v_t$
\STATE $C_s$ computes partial decryption $PD(g(v))$ and $PD'(h(v))$, and then  sends them to $C$, the controller of $v$.
\STATE $C$ partially decrypts $PD(g(v))$ and $PD'(h(v))$ and gets the plaintext of $g(v)$ and $h(v)$: $dg(v)$ and $dh(v)$.
\STATE $v_t=v_t$
\IF{$v_t \not \in S$}
	\STATE Let $S_t$ be the significant node set of $D_t$
	\STATE $v_{min}=-1,d_{min}=\infty$
	\FORALL{$v \in S_t$}
		\IF{$dg(v) + d(vv_t) < d_{min}$}
			\STATE $d_{min}=dg(v) + d(vv_t)$, $v_{min} = v$
		\ENDIF
	\ENDFOR
	\STATE Add the intra-domain path from $v_{min}$ to $v_t$ to $P$.
	\STATE Let $v_t=v_{min}$
\ENDIF
\STATE Now we construct the path from $v_s$ to $v_t$.
\WHILE{$v_t \neq v_s$}
	\IF{$dh(v_t) \in S_t$
	\COMMENT{$dh(v_t) \sim v_t$ is an intra-domain link}
	}
		\STATE Add the intra-domain path from $dh(v_t)$ to $v_t$ to $P$.
	\ENDIF
	\IF{$dh(v_t) \not \in S_t$
	\COMMENT{$dh(v_t) \sim v_t$ is an inter-domain link}
	}
		\STATE Add $dh(v_t) \sim v_t$ to $P$.
	\ENDIF
	\STATE Let $v_t=dh(v_t)$
\ENDWHILE
\end{algorithmic}
\end{algorithm}

\subsection{Implementation of Secure-If Operations}
\label{sec:if}

In this section, we will introduce the implementation of the three Secure-If operations used in the PSPT construction protocol.

First, we present a sketch of the Secure-If operation (See Algorithm~\ref{alg:secifsketch}).
Each Secure-If operation needs to construct three parameters $(t_0,t_1,t_2)$ and a condition satisfied value $x$ as input. $t_0$ is a condition while $t_1$ and $t_2$ are  two options. The output of Secure-If is $t_1$ when condition is satisfied ($t_0=x$); otherwise, the output is $t_2$. Such operation is achieved by an interactive process between two controllers, say $C_s$ and $C_i$. $C_s$ first applies partial decryption to $t_0$ and sends the result $PD(t_0)$, together with $t_1$ and $t_2$, to any other domain controller $C_i$. Then $C_i$ can fully decrypt $t_0$ and get the plaintext $dt_0$ as the threshold of secret sharing is $2$. $C_i$ verifies whether $dt_0$ is equal to $x$ and replies one of $t_1$ and $t_2$ (with re-randomization) to $C_s$.
With the Secure-If operation sketch, we need to show the construction of $(t_0,t_1,t_2)$ and $x$ when we introduce a Secure-If operation.

The PSPT construction uses three Secure-If operations ($SecIf_0$, $SecIf_1$, and $SecIf_2$).
As the Secure-If operation $SecIf_2$ is used in $SecIf_1$, our decryption is in the order of $SecIf_0$, $SecIf_2$, and $SecIf_1$.

\begin{algorithm}[t]
\caption{Secure-If Operation Sketch} \label{alg:secifsketch}
\begin{algorithmic}[1]
\REQUIRE ~\\$x$: value when condition is satisfied;\\
		 $(t_0,t_1,t_2)$: three parameters.\\

\STATE $C_s$ randomly choose a domain controller $C_i$.
\STATE $C_s$ computes $PD(t_0)$ and sends $(PD(t_0),t_1,t_2)$ to $C_i$. \COMMENT{$PD()$ is partial decryption operation}
\STATE Upon receiving $\left(PD(t_0),t_1,t_2\right)$, $C_i$ do partial decryption on $PD(t_0)$ and gets the plaintext $dt_0$ of $t_0$.
\STATE $C_i$ sends
		$\begin{cases}
			R(t_1) & if~dt_0==x \\
			R(t_2) & otherwise
		\end{cases}$
		back to $C_s$.
\STATE $C_s$ gets the result.
\end{algorithmic}
\end{algorithm}

\textbf{Construction of $SecIf_0$}

$x$ in $SecIf_0$ is $1$ and $(t_0, t_1, t_2)$ are constructed as follow.

With probability $\frac{1}{2}$, $C_s$ computes
$t_0=(\frac{f(v)}{f(v')})^r$, where $r$ is a randomly picked exponent
\footnote{Assume the plaintext space and the ciphertext space are both the same cyclic group. The value of $r$ needs to be picked uniformly at random from between $0$ and the order of the group minus $1$, including the two endpoints.}; $t_1= E(c_{max}|S|+1)$; $t_2=R(g(v)+g(v')+e(vv'))$. In this case if $f(v)$ is equal to $f(v')$, $t_0 = 1 = x$, hence the function of  $SecIf_0$ can be achieved.

With the remaining probability $\frac{1}{2}$, $C_s$ computes
$t_0=(\frac{1}{f(v)f(v')})^r$, where $r$ is a randomly picked exponent;
$t_1=R(g(v)+g(v')+e(vv'))$; $t_2= E(c_{max}|S|+1)$. In this case if $f(v)$ is not equal to $f(v')$, $t_0 = \frac{1}{2*1/2} = 1= x$, hence the function of  $SecIf_0$ can also be achieved.

The reason for that we use an uncertain calculation is to protect privacy. If we only apply the first case, any attacker that decrypts $t_0$ and finds $t_0=x$ can determine that $f(v) = f(v')$. However, in the current implementation, even if an attacker knows $t_0=x$, it cannot guess whether $f(v) = f(v')$ as $f(v) = f(v')$ and $f(v) \neq f(v')$ have equal probability.

%
%

\textbf{Construction of $SecIf_2$}


$x$ in $SecIf_2$ is $2$ and $(t_0, t_1, t_2)$ are constructed as follow.

With probability $\frac{1}{2}$, $C_s$ computes
$t_0=R(f(v))$. Let
$t_1, t_2$ be
$E'(2^{-1}), R'(f(v))$ for the $SecIf_2$ of $f(v)$;
$R(\alpha(vv')), R(g(v))$ for $g(v)$;
$E'(v'), R'(h(v))$ for $h(v)$;
$R'(f(v')), E'(2^{-1})$ for $f(v')$;
$R(g(v')), R(\alpha(vv'))$ for $g(v')$;
$R'(h(v')), E'(v)$ for $h(v')$.

With the remaining probability $\frac{1}{2}$,
$t_0=R(\frac{1}{f(v)})$ and the above values of  $t_1$ and $t_2$ are swapped, i.e., $t_1, t_2$ be
$R'(f(v)), E'(2^{-1})$ for $f(v)$ and so on;

\textbf{Construction of ${SecIf}_1$}

Here we show the construction of $x$ and $(t_0,t_1,t_2)$ in $SecIf_1$. We first introduce a comparison protocol called $osc$ which is necessary in $SecIf_1$.

The comparison protocol is designed by us based on the secure comparison protocol proposed in \cite{nishide2007multiparty}.
The protocol in \cite{nishide2007multiparty} takes
two ciphertexts of $E()$ as input, and outputs another ciphertext
of $E()$. The output is $E(1)$ if the first input's plaintext is greater than
or equal to the second input's; otherwise, the output is $E(-1)$.
Based on this comparison protocol, we design a new comparison protocol which can distinguish not only two edges with different $\alpha$ but also two edges with the same $\alpha$ by comparing their indexes.
Denote the original comparison operation by $sc()$.
Assume that the two edges' $\alpha$ values are $a$ and $b$ and their indexes are $a_{idx}$ and $b_{idx}$.

The protocol we designed, $osc(a,a_{idx},b,b_{idx})$, is actually a Secure-If operation. Its $x$ is $1$ and $(t_0,t_1,t_2)$ are constructed as the following paragraph. With $x$, $(t_0,t_1,t_2)$ and Secure-If operation sketch, we get the new protocol $osc$.

First we compute $\theta=sc(a,b) + sc(b,a) - E(1)$.
If $a\neq b$, $\theta$ is $E(-1)$; Otherwise $\theta$ is $E(1)$.
With probability $\frac{1}{2}$, $C_s$ computes
$t_0=\theta$; $t_1$ is $E(1)$ if $a_{idx} < b_{idx}$;
Otherwise $t_1$ is $E(-1)$; $t2=sc(b,a)$. With probability $\frac{1}{2}$, $C_s$ computes
$t_0=-\theta$; $t1=sc(b,a)$; $t_2$ is $E(1)$ if $a_{idx} < b_{idx}$;
Otherwise $t_2$ is $E(-1)$.

With the secure comparison, $C_s$ can compare each $\alpha$ value (except $\alpha(vv')$ itself) with $\alpha(vv')$.
Denote by $\beta_i$ the output of the protocol. Suppose that there
are $\zeta$ such outputs in total.
$C_s$ computes $\gamma=\sum_{i} \beta_i$, and uses the secure comparison
protocol again, to compare $\gamma$ with $E(\zeta)$.
Let $\epsilon$ be the output. With $\epsilon$, we can easily construct $t_0,t_1,t_2$ of $SecIf_1$.

The construction of $SecIf_1$ is shown in Algorithm~\ref{alg:secif1con}.

\begin{algorithm}[!ht]
\caption{Construction of $SecIf_1$} \label{alg:secif1con}
\begin{algorithmic}[1]
\ENSURE $x$ and $(t_0,t_1,t_2)$.
\STATE Denote by $m$ the total link number in the equivalent cost graph.\COMMENT{the link number is equal to the number of all $\alpha$ values.}
\STATE Assume the index of $\alpha(vv')$ is $k$.
\STATE $\gamma=E(0)$,$\zeta=m-1$
\FOR {$i=1$ \TO $m$}
	\STATE Assume the $i$th link is $v_1 \sim v_2$.
	\IF{$k \neq i$\COMMENT{$v1 \sim v_2 \neq v \sim v'$}}
		\STATE $\gamma = \gamma + osc(\alpha(vv'),k,\alpha(v_1v_2),i)$.
	\ENDIF
\ENDFOR
\STATE $\epsilon = osc(\gamma, E(\zeta))$.
\STATE Let $t_a$ be the result of $SecIf_2$.
\STATE Let $t_b$ be :\\ \hspace{0.2in}
	$\begin{array}{ccc}
	t_b=<&R'(f(v)),R(g(v)),R'(h(v)),&\\
	&R'(f(v')),R(g(v')),R'(h(v'))&>
	\end{array}$
\STATE $C_s$ computes: \\ \hspace{0.2in}
 $\begin{cases}
		t_0=\epsilon, t_1=t_a, t_2=t_b; with~probability~\frac{1}{2}, \\
		t_0=-\epsilon, t_1=t_b, t_2=t_a; with~probability~\frac{1}{2}
 \end{cases}$
\STATE $x=1$.
\STATE $C_s$ gets $x$ and $(t_0,t_1,t_2)$.
\end{algorithmic}
\end{algorithm}

%


%% file: 5_optimized.tex
\vspace{-1ex}
\section{Protocol Optimization}
\label{sec:opti}

In this section, we introduce two optimization methods of PYCRO. The first method reduces the number of shortest path tree computing for different flows. The second method reduces the computing time of each shortest path tree, called PYCRO with Candidate Recommendation (PYCRO-CR). Combining these two, the efficiency of PYCRO can be significantly improved.

\subsection{Shared Shortest Path Tree}
\label{sec:shared}
For all flows transmitted from a domain $D_s$, it is highly possible that another domain will treat these flows or a large subset of these flows using the same access and routing policy. We define an \emph{equal-flow group} $G$ as a group of flows from the same domain such that for all flows in  $G$, any other domain $D$ will treat them using the same  access and routing policy and hence provide the same paths between any two gateways of $D$. Therefore all flows in $G$ can use a number of shared shortest path trees.

A source domain with $k$ gateway switches maintains $k$ shared shortest path trees for each equal-flow group. Each of the trees is rooted at a gateway switch. To compute each shared shortest path tree, the nodes of the equivalent cost graph include gateway switches (significant nodes) of all domains.
Correspondingly, when constructing links in the equivalent cost graph, for any two significant
nodes $v$ and $v'$ ($v \neq v'$), we distinguish two cases:
\begin{itemize}
\item Case 1: $v$ and $v'$ belong to the same domain,
then there is a link in the equivalent
cost graph between those two nodes, and the cost of this link is $d(vv')$, which is only known to the domain of  $v$ and $v'$.
\item Case 2: $v$ and $v'$ belong to different domains,
If there is an inter-domain link
between $v$ and $v'$,
then there is a link in the equivalent
cost graph between those two nodes, and the cost of this link is $c(vv')$. If there is no
such inter-domain link
then there is no
link in the equivalent cost graph between these
two nodes.
\end{itemize}
Then the algorithm in Section \ref{sec:tree} can be run to build each shared shortest path tree.

When the source controller receives a flow query from the source $v_s$ to destination $v_t$. For each  gateway switch $v_i$ in the source domain, the source controller $C_s$ computes the encrypted distance from $v_s$ to $v_i$ plus the distance from $v_i$ to a gateway $w$ in the destination domain on the shared tree rooted at $v_i$. Thus for any $w$, there are $k$ potential paths from $v_s$ to $w$. Suppose the destination domain has $k'$ gateways. Then $C_s$ simply sends all $k\cdot k'$ path lengths, with partial encryption, to the destination controller $C_t$. $C_t$ can determine the shortest path from $v_s$ to $v_t$ and install forwarding entries using the method similar to that in Section \ref{sec:path}.

For example in Figure \ref{fig:ECG}, both $D_s$ and $D_t$ have two gateways. Hence for a group of flows from $D_s$, $D_s$ can maintains two PSPTs rooted at $v_1$ and $v_2$. There are at most $2\times2=4$ shortest paths between $D_s$ and $D_t$, and $D_s$ can select one of them for each flow.
Due to space limit, we do not present further details of path selection and forwarding entry installation in other domains.

\subsection{PYCRO with Candidate Recommendation}
\label{sec:PYCRO-CR}

The complexity of the shortest path tree algorithm presented in Section \ref{sec:tree} is mainly due to the number of calls of Secure-If operations to select the smallest $\alpha(vv')$ among the
$\alpha$ values of all links in the equivalent cost graph and the inefficiency of the secure comparison operation. To reduce the number of calls of Secure-If operations, we propose to use candidate recommendation to let the other domain recommend potential nodes that may have the smallest $\alpha$ value(i.e.,the smallest $\alpha$ value in its domain). As for the inefficiency of the secure comparison operation, we replace it with the Damgard-Geisler-Kroigard (DGK) secure comparison protocol, a more efficient protocol proposed in \cite{DamgardSG}.{\footnote{Using DGK, we make a small sacrifice of privacy for efficiency. However, it's worth since only a little information is revealed.}} Unlike the secure comparison we used in Section~\ref{sec:tree}, the input and output of the DGK protocol are plaintexts. Suppose there are two parties $A$ and $B$. $A$ has a number $a$ and $B$ has a number $b$. $A$ and $B$ can run the DGK protocol to compare $a$ and $b$ without revealing $a$($b$) to party $B$($A$).

After constructing the equivalent cost graph and adding the source node $v_s$ into the shortest path tree with $g(v_s)=E(0)$. The source domain controller $C_s$ broadcasts $v_s$ and $g(v_s)$ to all other domain controllers.
Then, the domains repeat the three interactive steps below for $|S|-1$ times:

Step 1.
Each domain $D_i$ finds its significant node that is not in the shortest path tree and the path length to the root is the shortest in $D_i$. $D_i$ also records the node's parent and its path length. We call the node selected by $D_i$ a \emph{candidate node} $v_i$.
Besides, a domain controller $C_0$ (specified by the source controller $C_s$) sends the information $g(v_0)$ and $h(v_0)$  of its candidate node $v_0$ to $C_s$.

Step 2. The source controller $C_s$ should then find out the candidate node whose path length to $v_s$ is the shortest.  $C_s$ temporarily sets $v_0$ as the shortest-distance node $u\leftarrow v_0$.
For each candidate node $v_i$ except $v_0$:
Controller $C_s$ sends a message including $g(u)$  to $v_i$'s controller $C_i$. $C_i$ then runs DGK secure comparison protocol to compare $g(u)$ and the path length of candidate node $v_i$. Once the DGK protocol finishes, $C_i$ tells $C_s$ the result of the comparison.
According to the result, if the plaintext of $g(u)$  is less than that of $g(v_i)$, $C_i$ then updates $u\leftarrow v_i$.

Step 3. After the two steps above, the controller $C_s$ get the shortest-distance node $u$. Next, $C_s$ requests the controller of $u$'s domain for the information of $g(u)$ and $h(u)$ and  add $u$ into the shortest path tree under its parent. $C_s$ broadcasts the new shortest path tree with encrypted distance information to the other domains.

After $|S|-1$ iterations of the above loop, $C_s$ finishes the computing of the shortest path tree.



%% file: 7_capacity.tex
\section{Bandwidth Allocation}
\label{sec:ba}
Bandwidth allocation has been applied to practical traffic engineering solutions such as B4 \cite{B4}. We solves a relatively simple version of the bandwidth allocation problem. Before we define the problem, we introduce some preliminaries.

\iftrue
Besides the link cost, every link $vv'$ also has a bandwidth $b(v,v')$. $b(v,v')$ represents the maximum bandwidth that link $vv'$ can provide.  And the definition of the cost of a flow on a path is:

\begin{define}
\label{def:1}
Given a path $p$ from node $v$ to node $v'$ whose length is $l_p$, if a flow $f$ consumes bandwidth $b_f$ on $p$, then the cost of $f$ on $p$ is  is $c(f, p) = b_f \cdot l_p$.
\end{define}

A flow $f$ has a bandwidth demand $q_f$. However as link bandwidth is limited, it may need multiple paths to satisfy a flow's bandwidth demand \cite{B4}. We assume a flow can be split to multiple subflows to be transmitted on different paths.
And the cost of $f$ is defined as:

\begin{define}
\label{def:2}
The cost of  $f$ for bandwidth allocation is the sum of path cost $\Sigma b_f \cdot l_P$ for $p \in P$ where $P$ is the set of paths that $f$ is split on.
\end{define}

\fi

Given Definition~\ref{def:1} and Definition~\ref{def:2},we define the \emph{Bandwidth Allocation} problem as follows:
\begin{define}
\label{def:ba}
\emph{Bandwidth Allocation}:for any flow $f$ with bandwidth demand $q_f$, we should find $k$ paths such that the sum of allocated bandwidth of these paths to $f$ is no less than the bandwidth demand $q_f$ and the routing cost of $f$ should  be minimized.
\end{define}

We design a bandwidth allocation protocol of PYCRO, named PYCRO-BA,  works in the following steps:

Step 1. During the construction of the equivalent cost graph, each domain controller assigns an available bandwidth  $b(v, v')$ amount between two significant nodes $v$ and $v'$, which is also encrypted by a homomorphic encryption system.

Step 2. The source controller creates the shortest path tree and finds the shortest path $p$ from the source $v_s$ to destination $v_t$ using the protocol presented earlier.

Step 3. The source controller determines the available bandwidth $b_p$ on the shortest path, which is the minimum value of $b(v, v')$ for all links $(v, v')$ on the paths. This process is similar to the previous protocol to determine the minimum cost candidate. We skip the protocol details here and have implemented them in the experiments.

Step 4. If $b_p$ is smaller than the bandwidth demand $q$, $C_s$ computes a residual demand $q- b_p$ and find another path to satisfy the demand.

Step 5. $C_s$ deletes all links of $p$ from the equivalent cost graph, and repeats Steps 2-4 to find more paths until the bandwidth demand is satisfied.

The above bandwidth allocation protocol requires multiple calls of the shortest path tree protocol. To improve its efficiency, $C_s$ may find multiple disjoint paths to different gateways of the destination domain and suggest these paths to the destination controller $C_t$. If $C_t$ can also find multiple disjoint paths from different gateways to $v_t$, multiple paths can be established by a single call of the shortest path tree protocol. We plan to develop more sophisticated protocol to optimize this process in future work.

%% file: 6_analysis.tex
\section{Privacy analysis of PYCRO}
\label{sec:analysis}

We analyze the privacy-preserving property of PYCRO in a standard cryptographic model,
the semihonest model \cite{Goldreich}, which is widely used in the literature (e.g., \cite{ZhongSDM} and \cite{PPGraph}). In
this model, all involved parties are assumed to follow the protocol faithfully, although
they may attempt to violate privacy using the information they obtain. Note that such
an assumption is acceptable in our scenario of cross-domain routing, because domain
controllers usually have long-term relationship with each other. Despite their curiosity
about others' private information, it is uncommon for them to deviate from the protocol
just in order to violate others' privacy.

The main result  we get£¬ as shown in Proposition 4 below, is that PYCRO only leaks to
each domain controller its significant nodes' distances from the source node and parents
nodes in the shortest path tree. We stress that this leaked distance information is
about a small number of pairs of nodes only. Any other information, including distances
between other pairs of nodes, are protected by PYCRO. Furthermore, our protection is
cryptographically strong, i.e., no partial knowledge about the protected information
is leaked by PYCRO. In contrast, the performance cost we pay for the privacy protection
is very reasonable. The execution time varies among different topologies, from seconds
to tens of seconds (please see Section \ref{sec:evaluation} for details).

\begin{proposition}
PYCRO is weakly privacy preserving in the semihonest model, in the sense that
it reveals to each domain controller no more than its significant nodes' distances from
the source node and parent nodes in the shortest path tree.
\end{proposition}

The basic idea of our proof is to demonstrate a probabilistic polynomial-time simulator according to the definition and proof methodologies of cryptographic protocols discussed in \cite{Goldreich}.

\begin{proof1} Due to limit of space, we only provide a proof sketch. Some details are skipped.

Our proof is established by demonstrating a probabilistic polynomial-time simulator according to the definition and proof methodologies of cryptographic protocols discussed in \cite{Goldreich}.

For each domain controller $C_i$, we construct a simulator for its view, which takes
as input  its significant nodes' distances from
the source node and parent nodes in the shortest path tree. All coin flips in the view
can be easily simulated, and thus we focus on generating simulated messages below.

If $C_i \neq C_s$, the simulator simulates the messages received from $C_s$ for
each of its significant node, using two  ciphertexts.
The first ciphertext is an encrypted distance of the significant node from the source node,
where the cryptosystem used is $E()$ and the key used is $C_i$'s own public key.
The second ciphertext is an encrypted identity of the significant node's parent node in
the shortest path tree, where the cryptosystem used is $E'()$ and the key used is still
$C_i$'s public key.

For $C_1$, we add the following simulated messages. In the Secure-If operation $SecIf_0$, the messages from $C_s$ is simulated using three ciphertexts. The first of
these three is under $E'()$, with the plaintext being $1$ with probability $\frac{1}{2}$, or
a uniformly random number with probability $\frac{1}{2}$. The remaining two
are encryptions of random plaintexts under $E()$. The public key used for encryption
of all these three is $C_1$'s own public key.

For the Secure-If operation $SecIf_1$ and $SecIf_2$, the simulator goes as follows. For $SecIf_2$,
the messages from $C_s$ are simulated using $8$ random ciphertexts under
$E'()$ and $4$ random ciphertexts under $E()$, and also another ciphertext under $E'()$
with the plaintext being $2$ or $\frac{1}{2}$, each with probability $\frac{1}{2}$, where
the public key used for encryption is $C_1$'s own public key.  For $SecIf_1$,
in addition to simulating the received messages in the executions of secure comparison,
the simulator simulates the earlier round of message from $C_s$ using three ciphertexts
under $E()$, with the first being an encrypted $1$ or encrypted $-1$, each
with probability $\frac{1}{2}$, where
the public key used for encryption is $C_1$'s own public key. The remaining two ciphertexts
are randomly generated.
The simulator simulates the later round of message from $C_s$ using $8$ random ciphertexts
under $E'()$ and $4$ random ciphertexts under $E()$,  and also another ciphertext
under $E()$ being an encrypted $1$ or encrypted $-1$, each
with probability $\frac{1}{2}$, where
the public key used for encryption is $C_1$'s own public key.

For $C_s$, the simulator goes as follows. First, it simulates the first round messages
from other domain controllers using random ciphertexts.  For each pair of significant
nodes in any other domain, there should be a random ciphertext under the cryptosystem
$E()$. Then the simulator proceeds to simulate the message received from $C_1$ in the
Secure-If operation $SecIf_0$. This should again be a random ciphertext under the cryptosystem
$E()$.

The Secure-If operation $SecIf_1$ and $SecIf_2$ are more complicated. For $SecIf_2$, the messages
from $C_1$ can be simulated by using $4$ random ciphertexts under cryptosystem $E'()$
and $2$ random ciphertexts under cryptosystem $E()$. For $SecIf_1$,  in addition
to simulating the received messages in the executions of secure comparison,
the simulator simulates the earlier message from $C_1$ using a random ciphertext, being
$E(1)$ with probability $\frac{1}{2}$ and $E(-1)$ with probability $\frac{1}{2}$. The
final messages from $C_1$ are simulated using $4$ random ciphertexts under $E'()$
and $2$ random ciphertexts under $E()$.
\end{proof1}

%% file: 8_evaluation.tex
\vspace{-4ex}
\section{Performance Evaluation}
\label{sec:evaluation}
The most significant concern of a privacy-preserving protocol is its computation and communication efficiency.
In this section, we conduct experiments to evaluate the efficiency of PYCRO protocols. We have implemented a prototype system on seven Dell PowerEdge R720 servers with Linux operation systems. All servers are connected via a campus network. Each machine runs a program to emulate a controller. If the controller number is larger than seven, we may run multiple threads on a single machine. We configure the controller placement such that two neighboring controllers are in different machines.
In all experiments, cryptographical operations are implemented using the Crypto++ library \cite{Dai}.

We use the router-level topologies of seven real ISP networks  collected by the Rocketfuel project \cite{Rocketfuel}.  The detailed information of the seven  networks can be found in Table~\ref{tbl:domaininfo} and networks are identified as I to VII. Based on topology analysis, we set a number of routers as gateways. Based on the seven networks, we construct 30 multi-domain topologies in six groups as shown in Table~\ref{tbl:testinfo}. For example, topologies 1 to 5  are constructed using the same domains I and II, but have different number of gateways and inter-domain links in an increasing order. Gateways are randomly selected from the gateways routers of the Rocketfuel networks. 

\begin{center}
\begin{table}[tp]
\renewcommand{\arraystretch}{1.3}
\caption{Information of the seven Rocketfuel topologies}
\label{tbl:domaininfo}
\centering
\scalebox{0.8}{
\begin{tabular}{c|c|ccc}
\hline
Network ID & Network name & \# routers & \# links & \# gateways \\
\hline
I & AS 1221 & $318$ & $758$ & $231$ \\
\hline
II & AS 1239 & $604$ & $2268$ & $242$ \\
\hline
III & AS 1755 & $172$ & $381$ & $61$ \\
\hline
IV & AS 2914 & $960$ & $2821$ & $507$ \\
\hline
V &AS 3257 & $240$ & $404$ & $89$ \\
\hline
VI &AS 3967 & $201$ & $434$ & $110$ \\
\hline
VII & AS 7018 & $631$ & $2078$ & $246$ \\
\hline
\end{tabular}}
\end{table}
\end{center}

\begin{center}
\begin{table}[tp]
\centering
\caption{Information of multi-domain topologies.}
\label{tbl:testinfo}
\scalebox{0.8}{\begin{tabular}{c|cccc}
\hline
Topo ID & \# domains & domains & \# inter-d links & \# gateways \\
\hline
$1-5$ & 2 & I,II & $10 - 100$ & $21-165$ \\
$6-10$ & 3 & I to III & $10-100$ & $21-158$ \\
$11-15$ & 4 & IV to VII & $10-100$ & $21-177$ \\
$16-20$ & 5 & I,III,V to VII & $10-100$ & $21-174$ \\
$21-25$ & 6 & I to VI & $10-100$ & $21-177$ \\
$26-30$ & 7 &I to VII & $10-100$ & $21-185$ \\
\hline
\end{tabular}}
\end{table}
\end{center}

\begin{figure*}[tp]
\centering
\begin{tabular}{p{120pt}p{120pt}p{120pt}p{120pt}}
\includegraphics[width=4.2cm]{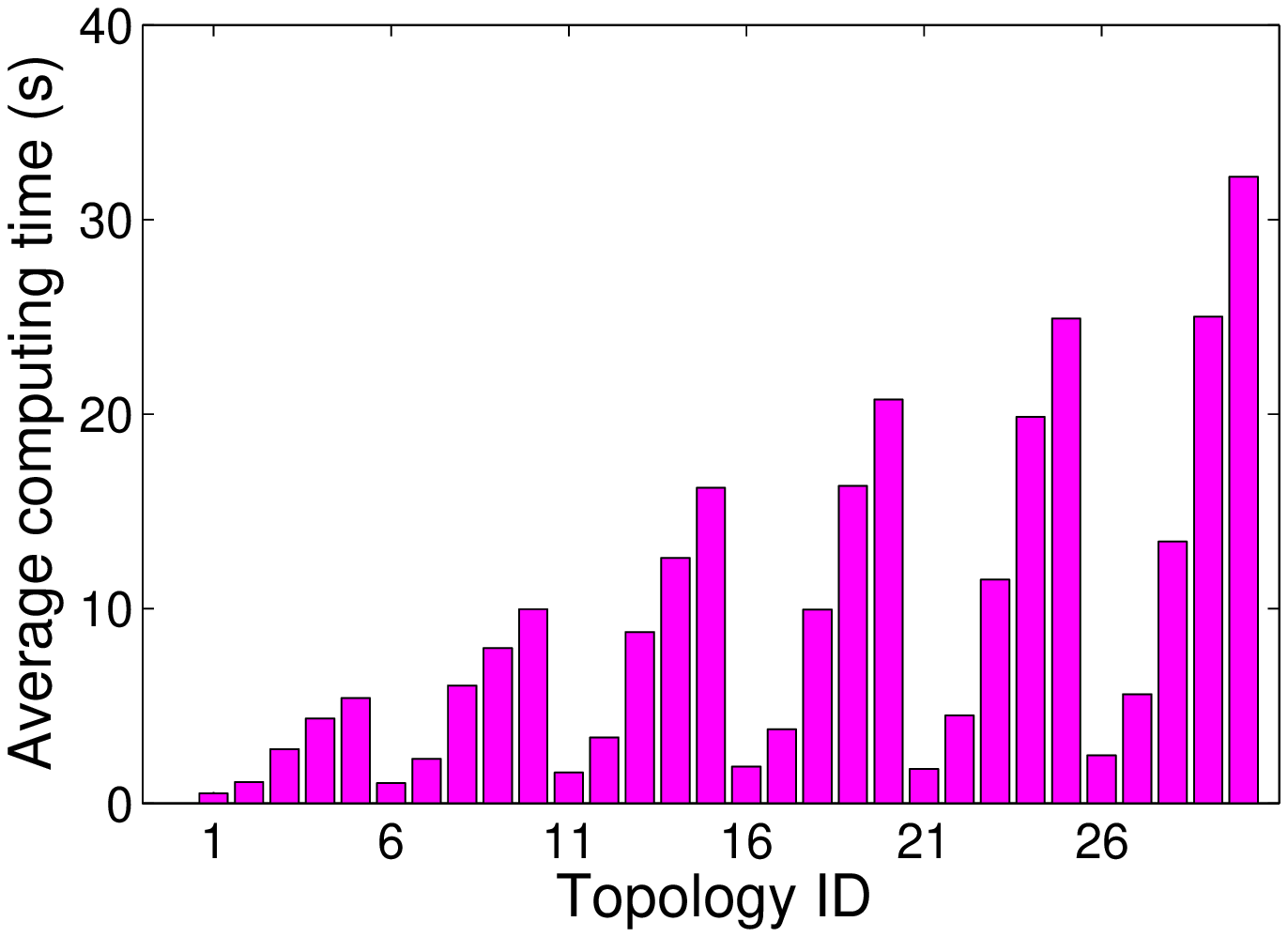}
\caption{Average  execution time  of PYCRO}
\label{fig:result 1}
&
\includegraphics[width=4.2cm]{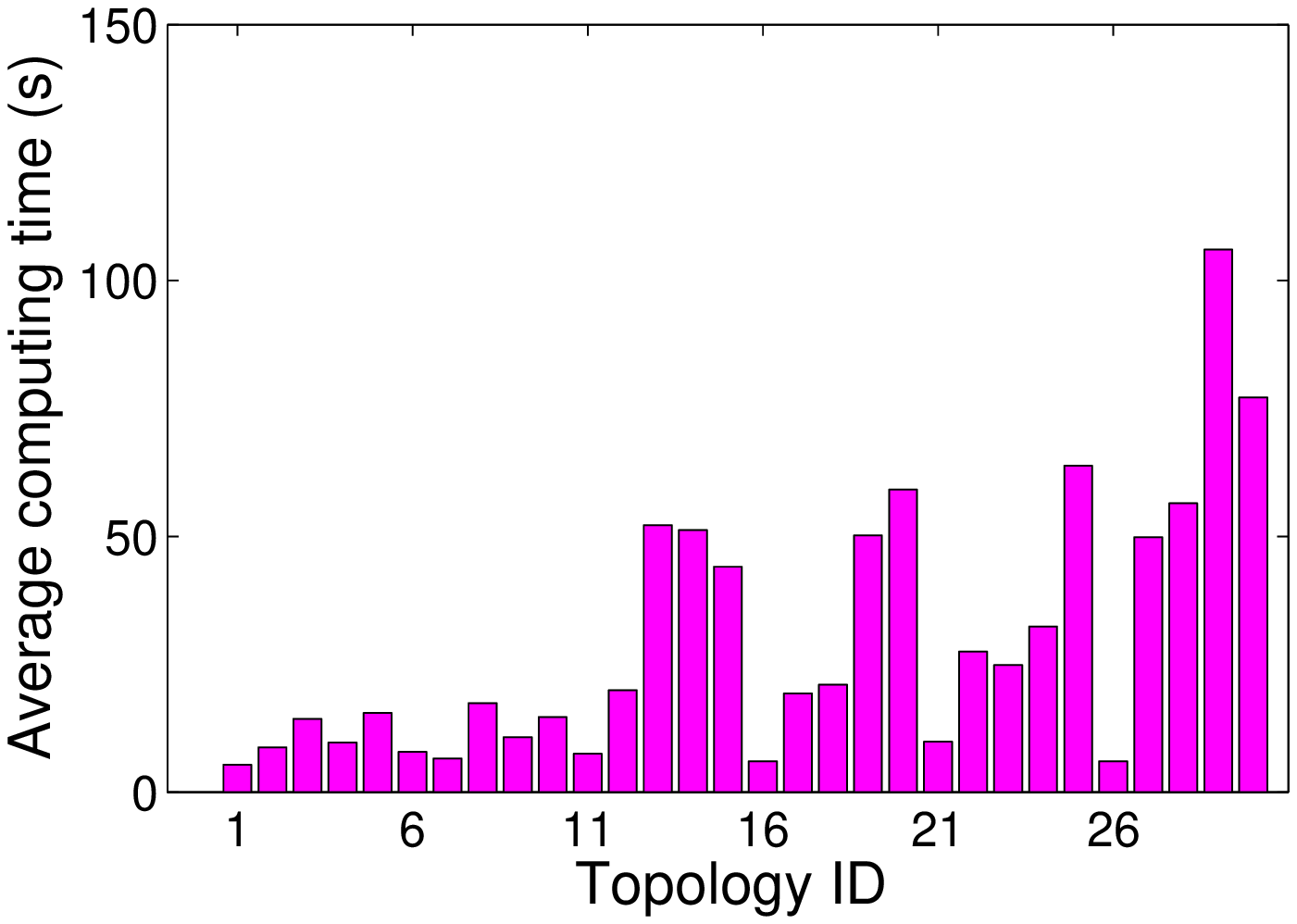}
\caption{Average  execution time  of PYCRO bandwidth allocation}
\label{fig:result 2}
&
\includegraphics[width=4.2cm]{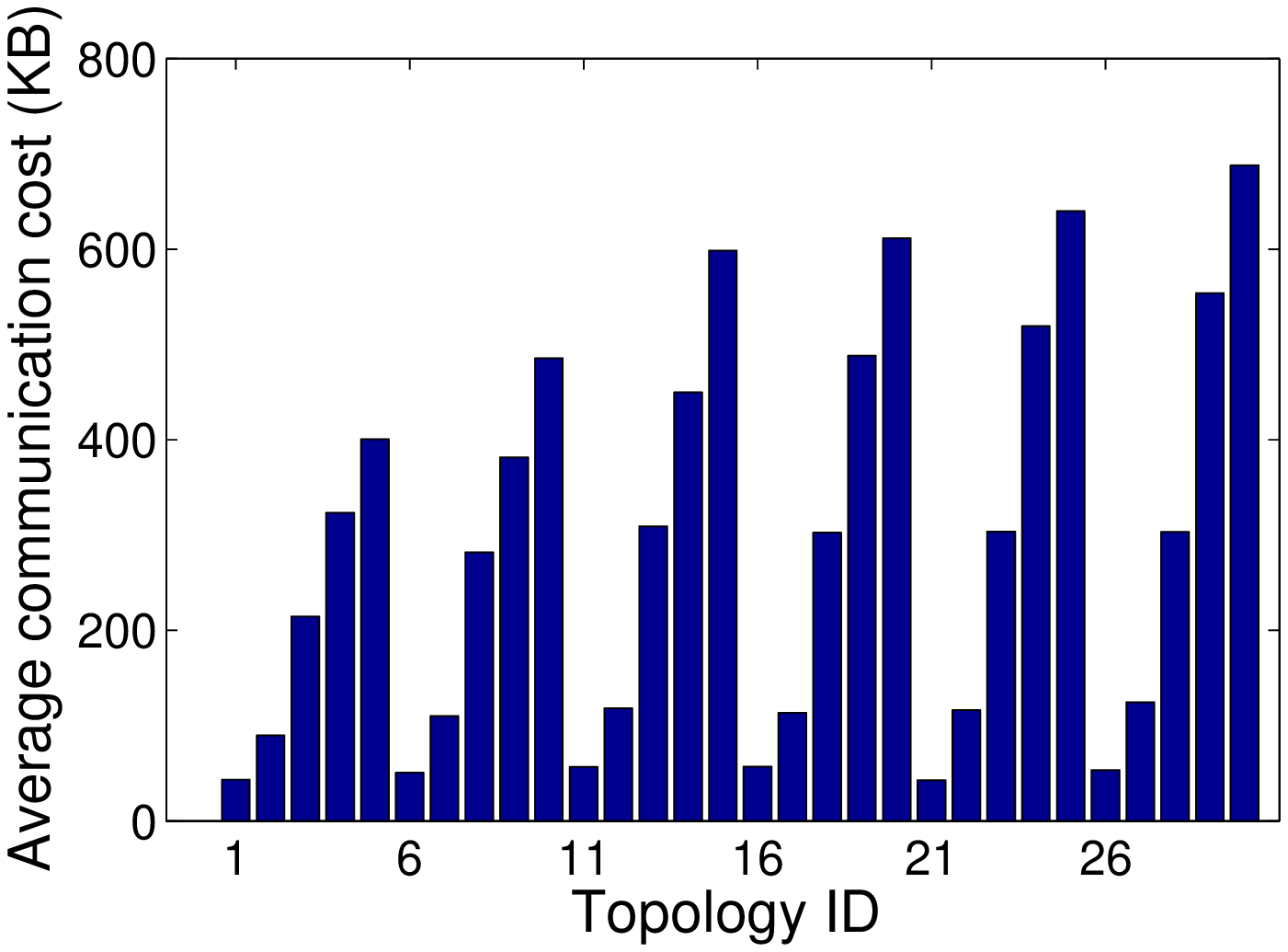}
\caption{Communication cost of PYCRO}
\label{fig:result 1_comm}&
\includegraphics[width=4.2cm]{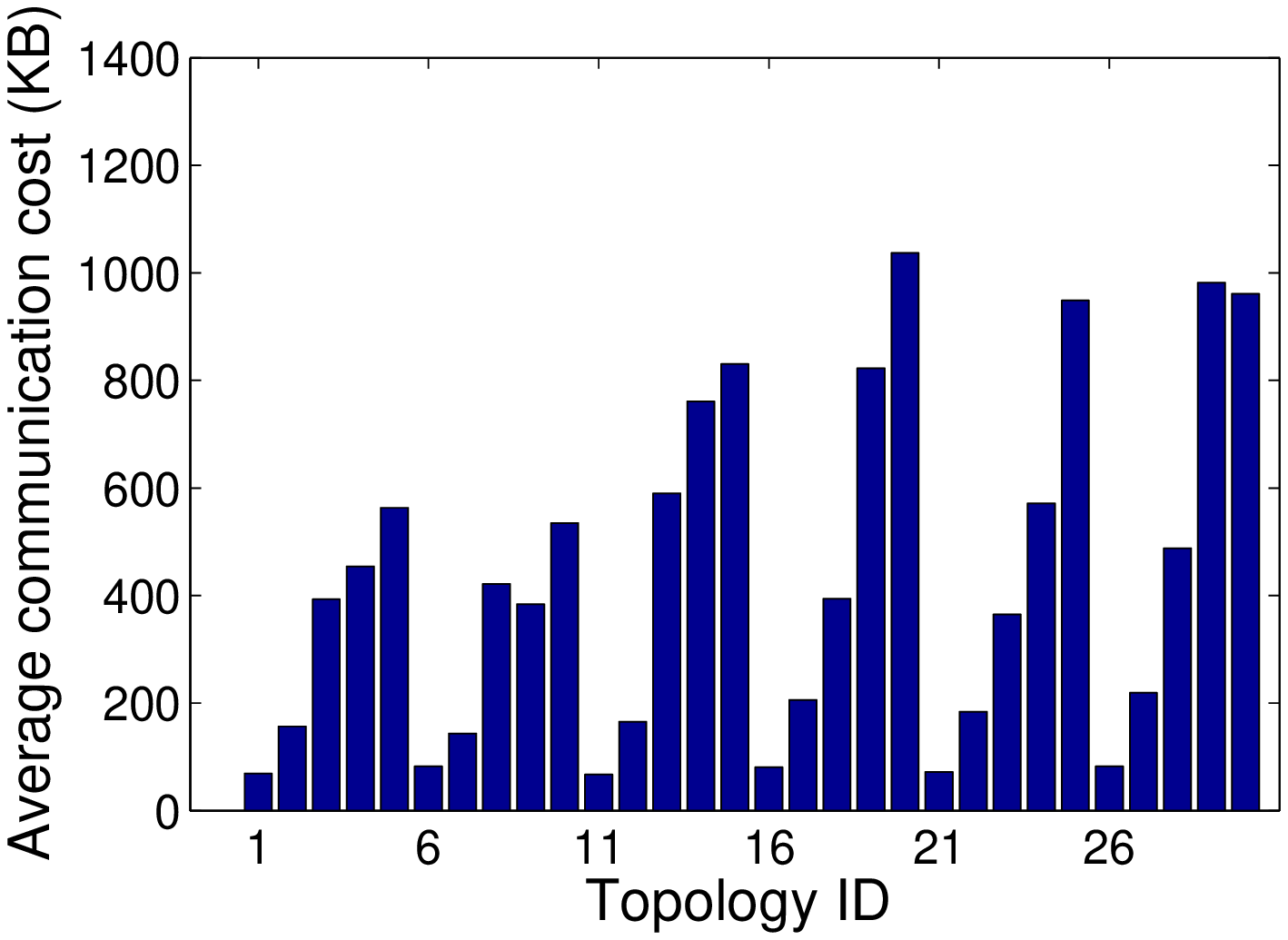}
\caption{Communication cost  of PYCRO bandwidth allocation}
\label{fig:result 2_comm}
\end{tabular}
\vspace{-3ex}
\end{figure*}

\vspace{-11ex}
\textbf{Computation cost.}
We first conduct experiments to construct shortest path trees on every topology. For each topology, we randomly select 20 nodes and construct a shortest path tree for each of them.
By computing time, we mean the average execution time of the protocol for one shortest path tree.
We find that the computing times for different nodes in a same topology vary very little.
It is because the execution time mainly depends on the number of domains, number of inter-domain links, and number of gateways.
Figure \ref{fig:result 1} shows the average execution time  of PYCRO on different topologies. The deviations are too small to be shown in the figure. We find that, for topologies consisting of the same domains (e.g., topologies 1-5), the  execution time increases linearly with the number of inter-domain links and number of gateways. By comparing topologies of different domains, the execution time also increases linearly with the number of domains. In general PYCRO is very efficient: it takes a short time to compute a shortest path tree on a topology with thousands of switches and links in a privacy-preserving manner.
Since a shortest path tree can be shared with multiple paths and the response to a path query takes much less time. Specially, if we have got a shortest path tree rooted at $v_s$, the paths that start from $v_s$ to any destination can be constructed easily and quickly using the Algorithm~\ref{alg:path_establishment}.



We then conduct experiments to evaluate the  execution time of the bandwidth allocation protocol. We assign every link a random capacity from $1$ to $5$. In each experiment, we set the bandwidth demand as 20 and find multiple paths between the sender and destination to satisfy the bandwidth demand. This bandwidth demand can be considered as the aggregated demand of all flows in the sender switch. For each topology we perform 20 runs and compute the average. The results are shown in Figure \ref{fig:result 2}. We find that there is no strict linear dependency of the execution time and number of inter-domain links, because more inter-domain links also make it easier to find multiple disjoint paths at a shortest path tree.

\textbf{Communication cost.}
We then show the communication cost of PYCRO in the average size of all messages per domain and plot the results in Figures \ref{fig:result 1_comm} and \ref{fig:result 2_comm}. We observe that the communication cost also increases with  the number of domains, number of inter-domain links, and number of gateways. Each domain spends less than 700 KB  to compute a shortest path tree  and less than 1 MB to allocate bandwidth for the largest topology. For other topologies the communication cost is much less. In general, PYCRO is communication efficient.

\textbf{Comparison with other solutions.}
It is hard to find an existing work achieving the same objectives as PYCRO. It is non-trivial to apply existing secure multi-party computation such as Fairplay~\cite{fairplay2004}
and SEPIA~\cite{sepia2010} to the problems of this paper.

A cross-domain privacy-preserving protocol for quantifying network reachability is proposed in~\cite{cross-reachability}. From their experimental results, we find that about $400$ or $550$ seconds offline computation cost, about $5$ or $25$ seconds online computation cost and about $450$ or $2100$ KB communication cost are needed for every party on average in their synthetic data.
In our experiments of optimized protocol of PYCRO, even the biggest network requires only $32.3/7=4.61$ seconds and $687.78/7=98.25$ KB for each domain in average.
In \cite{fairplay2004}, a full-fledged system called Fairplay that implements generic secure function evaluation is introduced.
Their experimental results show that it takes $1.41$ second to make a comparison.
In our optimized protocol, $(|S|-1)(n-1)$  comparison operations are needed in total,
where $|S|$ is the significant node number (from tens to hundreds) and $n$ is the domain number(from $2$ to $7$).
Hence, if we apply Fairplay to our protocol,
the average comparison operation time of each domain is $1.41(|S|-1)(n-1)$ seconds.
For a case $|S|= 185 $ and $n=7$, the average comparison time of each domain is $222.38$ seconds while the average time that PYCRO consumes in each domain is $4.61$ seconds. 

In summary, PYCRO can improve the time and bandwidth efficiency by an order of magnitude for cross-domain routing optimization, compared to existing solutions.